%% file: nhht.tex
\documentclass[reqno]{amsart}

\input{definitions.tex}
\gettimestamp Time-stamp: <2019-09-06 11:10:09 CDT>

\begin{document}

\title[Horseshoes]{Horseshoes for singly thermostated hamiltonians}
\author{Leo T. Butler}
\address{Department of Mathematics, University of Manitoba, Winnipeg,
  MB, Canada, R2J 2N2}
\email{leo.butler@umanitoba.ca}
\date{\timestamp}

\begin{abstract}
  This note studies 1 and 2 degree of freedom hamiltonian systems that
  are thermostated by a single-variable thermostat. Under certain
  conditions on the hamiltonian and thermostat, the existence of a
  horseshoe in the flow of the thermostated system is proven.
\end{abstract}
\begin{arxivabstract}
  This note studies 1 and 2 degree of freedom hamiltonian systems that
  are thermostated by a single-variable thermostat. Under certain
  conditions on the hamiltonian and thermostat, the existence of a
  horseshoe in the flow of the thermostated system is proven.
\end{arxivabstract}

\maketitle

\section{Introduction} \label{sec:intro}

One of the core models of equilibrium statistical mechanics is an
isolated mechanical system, modeled by a hamiltonian $H$, that is
immersed in, and in equilibrium with, a heat bath $B$ at the
temperature $T=1/β$. Nos{é} \cite{nose}, based on earlier work of
Andersen \cite{andersen}, created a dynamical model of the exchange of
energy between heat bath and system. This consists of adding an extra
degree of freedom $s$ and rescaling momentum by $s$:
\begin{align}
  \label{eq:nose}
  G &= H(q,p s^{-1}) + N(s,p_s), &&& \text{where }N(s,p_s) =\dfrac{1}{2 M} p_s^2 + n k T \ln s,
\end{align}
$n$ is the number of degrees of freedom of $H$, $M$ is the
``mass'' of the thermostat and $k$ is Boltzmann's constant. Solutions
to Hamilton's equations for $G$ model the evolution of the state of
the infinitesimal system along with the exchange of energy with the
heat bath.

Hoover reduced Nos{é}'s thermostat by eliminating the state variable
$s$ and rescaling time $t$~\cite{hoover}:
\begin{equation*}
  q = q,\qquad ρ = ps^{-1},\qquad \ddt{τ} = s \ddt{t},\qquad ξ = \ddt[s]{t}.
\end{equation*}
The Nos{é}-Hoover thermostat for a $1$ degree of freedom hamiltonian
$H$ can be put in the form:
\begin{align}
  \label{eq:nose-hoover}
  \dot{q} &= H_{ρ}, && \dot{ρ} = -H_q - ε \xi ρ, &&
  \dot{\xi} = ε \left( ρ · H_{ρ} - T \right),
\end{align}
where $ε ² = 1/M$ (see Lemma~\ref{lem:single-2-thermostat} below).

Hoover observed that this thermostat is ineffective in producing the
statistics of the Gibbs-Boltzmann distribution from single orbits of
the thermostated harmonic oscillator~\cite{hoover}. There are numerous
extensions of the Nos{é}-Hoover thermostat; a sample of these works
includes~\cite{MR1079786,MR1150101,nose-hoover-chains,MR2136523,sciversesciencedirect_elsevier0375-9601.95.00973-6,tayfranc10.1080/00268971003689923}. In~\cite{MR1079786,MR1150101},
a class of two-variable thermostats is introduced that is formally
similar to, and extends the Nos{é}-Hoover thermostat by controlling
both momentum and configuration
variables. In~\cite{nose-hoover-chains}, a class of $n$-variable
thermostats is introduced by making the simple observation that the
Nos{é} hamiltonian~\eqref{eq:nose-hoover} can itself be
thermostated--hence $n$ thermostats can be recursively chained
together. \cite{MR2136523} studies a few variants of recursively
thermostats. In~\cite{sciversesciencedirect_elsevier0375-9601.95.00973-6},
a two-variable thermostat is used to control the first two non-trivial
moments of momentum of the thermostated
system. In~\cite{tayfranc10.1080/00268971003689923}, the notion of a
chain of thermostats is extended to a network of thermostats which are
coupled based on a graph. All of these extensions share a common
feature: the state of the thermostat is $n ≥ 2$ dimensional.

On the other hand, there are numerous extensions of the Nos{é}-Hoover
thermostat that model the exchange of energy with the heat bath using
a single, additional thermostat variable ($ξ$ in
\ref{eq:nose-hoover}), the so-called \textem{single thermostats}. A
sample
includes~\cite{proquest1859245445,PhysRevE.75.040102,doi:10.1142/S0218127416501704,10.12921/cmst.2016.0000061,Wang2015,doi:10.1063/1.4937167,doi:10.1142/S0218127417501115}. In~\cite{proquest1859245445},
a ``variable'' mass thermostat is introduced, which precludes the
Hooverian
reduction~\eqref{eq:nose-hoover}. In~\cite{PhysRevE.75.040102}, singly
thermostated harmonic oscillators, where the thermostat controls a
single moment of momenta-which might be thought of as a weighted
average temperature-are shown to have a first-order averaged system
that is integrable. In~\cite{doi:10.1142/S0218127416501704}, a variant
of the Nos{é}-Hoover thermostated harmonic oscillator is considered,
where the total-as opposed to kinetic-energy is
controlled. In~\cite{10.12921/cmst.2016.0000061}, the linear friction
of the Nos{é}-Hoover thermostat is replaced with a $\tanh$-friction
that saturates at large magnitudes of the thermostat state
$ξ$. In~\cite{Wang2015,doi:10.1063/1.4937167}, the Nos{é}-Hoover
thermostated harmonic oscillator is investigated and regions of phase
space are found where apparently chaotic dynamics exist and regions
where invariant tori appear with various knot
types. In~\cite{doi:10.1142/S0218127417501115}, the same authors visit
the variant of the Nos{é}-Hoover thermostated harmonic oscillator that
thermostats total energy, and demonstrate (numerically) the existence
of a horseshoe. The numerical evidence also seems to show that the
knotted and linked tori that are found are due to secondary
bifurcations of KAM tori from the first-order averaged
system~\cite[fig. 2]{doi:10.1142/S0218127417501115}.

In addition to the extensions of the Nos{é}-Hoover thermostat, there
are several notable recent studies of this thermostat itself
including~\cite{MR2299758,MR2519685,Mahdi20111348}. In~\cite{MR2299758},
it is shown that the Nos{é}-Hoover thermostated harmonic oscillator
enjoys KAM tori near the $ε=0$ decoupled limit;~\cite{MR2519685}
extends these results and shows that a Nos{é}-Hoover thermostated
integrable system has a first-order averaged system that is also
integrable. In~\cite{Mahdi20111348}, it is shown that the
Nos{é}-Hoover thermostated harmonic oscillator is not integrable in
the class of Darboux integrals, which implies non-integrability in the
class of polynomial integrals, but not necessarily in the class of
real-analytic or smooth integrals.

\subsection{Results: $2$ degrees of freedom}
\label{sec:intro-results-2d}

The first result of the present note concerns the creation of
transverse homoclinic points by the Nos{é}-Hoover thermostat and
thermostats similar to it. To do this, I start with an integrable
$2$-degree of freedom hamiltonian system which enjoys a normally
hyperbolic invariant $2$-manifold foliated by periodic orbits and with
coincident stable and unstable manifolds. When thermostated, the
invariant manifold is ``blown up'' into a normally hyperbolic
invariant $4$-manifold that is foliated by isotropic
$2$-tori. \poincaremelnikov{} theory is applied to detect transverse
intersections of the stable and unstable manifolds. This theory is
perturbative in nature, so to obtain a system to which the theory is
applicable, a few steps are needed. First, the $6$-dimensional
symplectic phase space is symplecticly reduced by a $\T^1$
action. The resulting, parameterized family of $2$-degree of freedom
hamiltonians are nearly integrable and enjoy a normally hyperbolic
invariant $2$-manifold that is foliated by periodic orbits--i.e. each
system enjoys a saddle-centre fixed point. Second, the systems are
rescaled to simplify the parametric dependence, elucidate the nature
of the saddle-centre and the near-integrable nature of the
problem. This work is all done within a hamiltonian
framework. Surprisingly, the nature of the thermostat equations
appears to preclude the use of a hamiltonian formalism to compute the
\poincaremelnikov{} function that detects the transverse homoclinic
points. The penultimate step is a two-part rescaling that destroys the
canonical nature of the differential equations but reveals both the
near-integrability of said equations and makes computation of the
\poincaremelnikov{} function transparent.

Let us state the theorem:
\begin{theorem}
  \label{thm:intro-thm-2d}
  Let $r ≥ 3$, $Ξ = \R$ or $\T^1$ and $V : Ξ → \R$ be a $C^r$
  potential function that has a unique non-degenerate local maximum value of
  $0$ at $q=0$. Let
  \begin{align}
    \label{eq:intro-spec-ham}
    H(q_1,p_1,q_2,p_2) & = \overbrace{½ p_1 ²}^{H^{(1)}} + \overbrace{\frac{1}{2m} p_2 ² + V(q_2)}^{H_m^{(2)}} \\\notag
    N_T(s,S)           & = \frac{1}{2M} S ² + T \ln s + O(S ³),
  \end{align}
  where $(q_1,p_1) ∈ \cotangent \T^1$ and $(q_2,p_2) ∈ \cotangent
  Ξ$. Define $μ = p_1$, $1/α² = μ\sqrt{2M}$ and $β=1/T$.

  The subspace $N=\set{(q_1,p_1,q_2,p_2,s,S) \mid q_2=0, p_2=0}$ is a
  normally hyperbolic invariant manifold of the thermostated
  hamiltonian $G=H+N_T$~\eqref{eq:nose}. The reduction of the level
  set $p_1=μ$ by the $\T^1$ action reduces $N \cap \set{p_1=μ}$ to a
  normally hyperbolic invariant manifold $\reduction{N}_{μ}$ for the
  reduced hamiltonian $\reduction{G}_{μ}$.
  
  Let $\fouriertransform{f}$ be the Fourier transform of the function
  $t \mapsto p_2(t)^2$ along a homoclinic connection to $(0,0)$ for
  $H^{(2)}_m$. If $\fouriertransform{f}(α^2/π) ≠ 0$, then with
  $m = β/μ^2$ and $2M = 1/(α μ)^2$ as $β, 1/|μ|, m → 0$ the stable and
  unstable manifolds of $\reduction{N}_{μ}$ intersect transversely in
  a neighbourhood of the saddle-centre equilibrium
  $(q_2=0,p_2=0,s=√ T,S=0)$.
\end{theorem}

Note that $\fouriertransform{f}$ is the Fourier transform of the
kinetic energy in the subsystem described by $H^{(2)}_m$; since
$(0,0)$ is a non-degenerate saddle for this subsystem, the kinetic
energy decays exponentially as $t → ± ∞$. This implies that
$\fouriertransform{f}$ is real-analytic and has at most countably many
zeros.

It needs to be emphasized that Theorem~\ref{thm:intro-thm-2d} does not
simply apply to the Nos{é}-Hoover thermostat; it includes the logistic
thermostat of~\cite{10.12921/cmst.2016.0000061}, for example. Below,
it is also shown that the results extend with only minor modifications
to the variable-mass thermostats like that
in~\cite{proquest1859245445} (see
Theorem~\ref{thm:winkler-a1-pos-a1=0} below).

In addition, I should point out that Theorem~\ref{thm:intro-thm-2d}
admits a few straightforward extensions. One can take $H^{(1)}$ to be
a purely kinetic hamiltonian on $\cotangent \T^n$ for any $n ≥ 1$;
equally, it could be a bi-invariant metric hamiltonian on
$\cotangent G$ for any compact Lie group $G$. Once momentum is fixed
and the co-adjoint orbit of the momentum is symplecticly reduced by
the $G × G$ action on $\cotangent G$, the reduced equations are
essentially those in Theorem~\ref{thm:intro-thm-2d}. Moreover, one can
take $H^{(2)}_m$ to be a sum of any number of decoupled $1$-degree of
freedom mechanical hamiltonians each with a potential satisfying the
same condition. Finally, one can even take $H^{(1)}$ to be mechanical
with a non-trivial potential. By applying averaging to the $H^{(1)}$
subsystem at high energy, it behaves like a purely kinetic system up
to an error that is negligible for the purposes here.

\subsection{Results: $1$-degree of freedom}
\label{sec:intro-results-1d}

Much of the literature on single thermostats focuses on singly
thermostated $1$-degree of freedom hamiltonians. The relatively poor
``thermalization'' of the Nos{é}-Hoover thermostated harmonic
oscillator led researchers to introduce more non-linearity into the
thermostat. One way to do this, that has not been explored in the
literature, is to make the thermostat's state compact. To see why this
might be interesting, consider a thermostat friction, like the
logistic thermostat, that is odd in $ξ$ and which saturates at $1$
when $ξ=∞$. The planes $ξ=± ∞$ in the extended phase space possess
straight-line connecting orbits that connect $(q_c,0,-∞)$ to
$(q_c,0,+∞)$ when $q_c$ is a critical point of the potential
energy. If we glue two copies of the extended phase space along copies
of $ξ=+∞$ and $ξ=-∞$ with $\dot{ξ}$ reversed in one copy, the result
is a thermostated system with the thermostat friction depending
periodically on the state. The planes at $ξ=± ∞$ are invariant
manifolds, so they separate the two copies, but one can see that the
reversal of $\dot{ξ}$ means that this invariance can be destroyed by a
small perturbation. That is, the dynamics of the two systems can be
made to intermingle.

To state the result of this note in this direction, let $Σ$ be a
$1$-manifold, $\cotangent Σ = Σ × \R$ be its cotangent bundle and
$H : \cotangent Σ → \R$ be a $C^r$ hamiltonian, $r>2$. Let $Ξ$ be a
$1$-manifold and $P = \cotangent Σ × Ξ$ be a trivial $Ξ$-bundle over
$\cotangent Σ$ with projection maps
\[
  \xymatrix{Ξ & \ar@{->>}[l]_{ξ} P \ar@{->>}[r]^{π} & \cotangent Σ}.
\]
The Poisson bracket $\pb{}{}$ on $\cotangent Σ$ pulls back to $P$ in a
natural manner, as does the hamiltonian vector field
$X_H = \pb{\ }{H}$. A $C^{r-1}$ vector field $Τ$ on $P$ is a
\defn{thermostat} for $H$ if it satisfies the
definition~\ref{def:thermostat} below.

Let $s ∈ \cotangent Σ$ be a saddle critical point for $H$. Assume that
$γ ⊂ \cotangent Σ$ is a homoclinic connection for $s$ that bounds the
compact region $r$. Then $N = \set{s} × Ξ$ is a normally hyperbolic
invariant manifold for $X_H$ whose stable and unstable manifolds
$W^{±}(N)$ contain $Γ = γ × Ξ$ and $Γ$ bounds the region $R = r ×
Ξ$. By invariant manifold theory~\cite{MR0501173}, $Y_{ε}=X_H + ε Τ$
possesses a normally hyperbolic invariant manifold $N_{ε}$ that is a
graph over $N$ and similarly the local stable (unstable) manifold
$W_{loc}^{±}(N_{ε})$ is a graph over $W_{loc}^{±}(N)$. Let us say that
the thermostat $Τ$ is \defn{monotone} there is a neighbourhood $U$ of
$Γ=γ × Ξ$ such that $\kp{\D{ξ}}{Τ}$ does not vanish on $U$. This
implies that $Y_{ε} | N_{ε}$ does not vanish for all $ε$ in some
deleted neighbourhood of $0$. If $Τ$ is monotone, then $N_{ε}$ is an
orbit of $Y_{ε}$; when $Ξ=\T^1$, this orbit is periodic.

\begin{theorem}
  \label{thm:intro-thm-1d}
  Let $Σ=\R$ or $\T^1$, $Ξ=\T^1$ and $H : \cotangent Σ → \R$ be a
  $C^r$ hamiltonian, $r>2$, that has a saddle critical point $s$ with
  homoclinic connection $γ$. If $Τ$ is a $C^{r-1}$ monotone thermostat
  for $H$ that is topologically transverse (resp. transverse) at $γ$
  (definition \ref{def:cond-div}), then for all $ε ≠ 0$ sufficiently
  small, the stable and unstable manifolds of the periodic orbit
  $N_{ε}$ are topologically transverse (resp. transverse).

  In particular, the thermostated vector field $Y_{ε}$ enjoys a
  horseshoe.
\end{theorem}

\subsection{Outline}
\label{sec:outline}

This note is organized as follows: \S \ref{sec:pre-mat} gives a proof
of Theorem~\ref{thm:intro-thm-1d}; \S \ref{sec:horseshoes} gives a
proof Theorem \ref{thm:intro-thm-2d}; \S \ref{sec:winkler} proves an
extension of the latter to variable-mass thermostats; \S
\ref{sec:conclusion} concludes.

\section{Horseshoes in a thermostated $1$-degree of freedom hamiltonian}
\label{sec:pre-mat}

Let us use the notation and terminology in the paragraph preceding the
statement of \ref{thm:intro-thm-1d}. Let
$ω = π^* \left( \D{p} ∧ \D{q} \right)$ be the pullback of the
canonical symplectic form on $\cotangent Σ$ to $P$ and for each
homoclinic connection $γ$ to a saddle critical point $s$ of $H$ define
a function $M_{γ} : Ξ → \R$,
\begin{align}
  \label{eq:mel-fn}
  M_{γ}(ξ) &= \oint_{γ × \set{ξ}} ι_{Τ} ω, &&& \forall ξ ∈ Ξ.
\end{align}
If one writes $Τ = A\ ∂_q + B\ ∂_p + C\ ∂_{ξ}$, then
$ι_{Τ}ω = -A\ \D{p}+B\ \D{q}$ so $\D{} ι_{Τ}ω = \left( A_q + B_p
\right)\, \D{p} ∧ \D{q}$ and
$M_{γ}(ξ) = \oint_{γ × \set{ξ}} -A\ \D{p}+B\ \D{q}$ which equals
$\iint_{r × \set{ξ}} \left( A_q + B_p \right)\, \D{p} ∧ \D{q}$.

\begin{definition}
  \label{def:cond-div}
  $Τ$ is \defn{topologically transverse at $γ$} if $M_{γ}$ changes
  sign; it is \defn{transverse at $γ$} if it is topologically
  transverse and $0$ is a regular value of $M_{γ}$.
  If $Τ$ is topologically transverse (resp. transverse) at each $γ$,
  then $Τ$ is said to be topologically transverse (resp. transverse).
\end{definition}

In~\cite{1806.10198v3}, the following definition of a thermostat
vector field for a hamiltonian $H$ is introduced. It is intended to
capture the idea that the extended dynamics on the extended phase
space should preserve a Gibbs-Boltzmann type probability measure
$\D{μ}_{β}$ whose marginal over $ξ$ should be the original
Gibbs-Boltzmann measure. Moreover, the extended dynamics should heat
the system at low energy and cool it at high energy--at least on
average.

\begin{definition}
  \label{def:thermostat}
  A smooth vector field $Τ$ on $P$ is a \defn{thermostat} for $H$ if
  there is a smooth probability measure
  \begin{equation}
    \label{eq:dmu}
    \D{μ}_{β} = Z_1(β)^{-1} \exp(-β G_{β}(q,p,ξ) )\, \D q\, \D p\, \D ξ
  \end{equation}
  on $P$ such that the following holds
  \begin{enumerate}
  \item\label{def:thermostat-invariance} $\D{μ}_{β}$ is invariant for $Y_{ε} = X_H + ε Τ$ for all $ε$;
  \item\label{def:thermostat-proper} $G_{β}=G_{β}(H,ξ)$ is proper for all $β>0$;
  \item\label{def:thermostat-thermostat} there exists an interval of regular values of $H$,
    $[c_-,c_+]$, and constants $d_{±}$ such that
    \begin{enumerate}
    \item the average value of $\kp{\D ξ}{Τ}$ is of opposite sign
      on $H^{-1}(c_-) \cap ξ^{-1}([d_-,d_+])$
      and $H^{-1}(c_+) \cap ξ^{-1}([d_-,d_+])$;
    \item the average value of $\kp{\D H}{Τ}$ is of opposite sign
      on $ξ^{-1}(d_-) \cap H^{-1}([c_-,c_+])$
      and $ξ^{-1}(d_+) \cap H^{-1}([c_-,c_+])$.
    \end{enumerate}
  \end{enumerate}
\end{definition}

The average values in condition~\ref{def:thermostat-thermostat} are
orbit averages, taken over the orbits of the vector field $Y_0=X_H$.

\begin{proposition}
  \label{prop:top-transverse}
  If $Τ$ is a $C^{r-1}$ monotone thermostat that is
  (resp. topologically) transverse at $γ$, then for all $ε ≠ 0$
  sufficiently small, the first-return map $f_{ε}$ of the vector field
  $Y_{ε}=X_H+ ε Τ$ enjoys a horseshoe in a neighbourhood of
  $γ × \set{0}$.
\end{proposition}
\begin{proof}
  First, let us show that the local stable and unstable manifolds of
  $N_{ε}$, $W^{±}_{loc}(N_{ε})$, split.
  The proof is a straightforward application of the
  \poincaremelnikov{} function: $W^{±}_{loc}(N_{ε})$ is a graph over
  $W^{±}_{loc}(N)$ so there are sections
  $Ω^{±}_{ε} : W^{±}_{loc}(N) → W^{±}_{loc}(N_{ε})$ and
  $H \circ Ω^+_{ε} - H \circ Ω^-_{ε} = ε M + O(ε ²)$. The function
  $M(X) = \int_{-∞}^{∞} \kp{\D{H}}{Τ} \circ φ_0^t(X)\, \D{t}$
  where $φ^t_0$ is the flow map of $Y_0=X_H$ and $X ∈ W^{±}(N)$. It is
  straightforward to see that $M|_{Γ}=M_{γ}$.

  Next, let us show that there is a horseshoe. Let us fix $ε>0$ such
  that the local stable and unstable manifolds, $W^{±}_{loc}(N_{ε})$,
  split to first order, so the zeros of the \poincaremelnikov{}
  function $M$ detect the splitting. Assume that $M_{γ}(ξ)$ changes
  sign at $ξ=0$. Let $ζ : \cotangent Σ → \R$ be a smooth function that
  vanishes at the critical points of $H$ and $ζ |_{γ}$ has a
  non-degenerate zero. Let $Ζ_{η} = \set{(q,p,ξ) \mid ξ=η ζ(q,p)}$ be
  a smooth surface that intersects $γ × \set{0}$ transversely. Let
  $f_{ε} : U → Ζ_{η}$ be the first-return map from a neighbourhood
  $U ⊂ Ζ_{η}$ of $\set{(q,p,ξ) \mid ξ = η ζ (q,p), (q,p) ∈ γ}$ along
  the flow $φ^t_{ε}$ of $Y_{ε}$. For all $ε, η>0$ sufficiently small,
  such a neighbourhood and return map exist by virtue of the monotone
  condition. By construction, $f_{ε}$ preserves an area form since
  $φ^t_{ε}$ preserves the volume form $\D{μ}_{β}$; it has a hyperbolic
  fixed point $s_{ε}=s+O(ε)$ and the local stable and unstable
  manifolds of $s_{ε}$, $W^{±}_{loc}(s_{ε})$, coincide with
  $W^{±}_{loc}(N_{ε}) \cap Ζ_{η}$. Therefore, these manifolds
  split. If $Τ$ is transverse, then they split transversely and an
  application of the Birkhoff-Smale homoclinic theorem implies the
  result. Otherwise, if $Τ$ is only topologically transverse, the work
  of Burns \& Weiss implies the result~\cite{MR1346373}.
\end{proof}
\begin{remark}
  \label{rem:top-transverse}
  The reason that $r>2$ in Proposition~\ref{prop:top-transverse} is
  because Burns \& Weiss prove that if the diffeomorphism $f$
  ($=f_{ε}$ above) is $C^1$, then it possesses an invariant set $Λ$
  which factors onto a full shift on two symbols. If the
  diffeomorphism is area preserving and $C^r$ for $r>1$, then one can
  appeal to a theorem of Katok~\cite{MR573822} which states that in
  such a case, the diffeomorphism $f$ actually possesses a
  horseshoe. Hidden in the proof of
  Proposition~\ref{prop:top-transverse} is the fact that the return
  time to the cross-section is $O(1/ε)$.
\end{remark}

\begin{definition}
  \label{def:separable-thermostat}
  Let
  \[
    Τ = A\ ∂_q + B\ ∂_p + C\ ∂_{ξ}
  \]
  be a $C^1$ vector field on $P$. $Τ$ is \defn{separable} if
  \begin{enumerate}
  \item $A = A_0(q,p)\, A_1(ξ)$, and similarly for $B$ \& $C$;
  \item $A_1$ \& $B_1$ are odd functions; and
  \item $A^2 + B^2 > 0$ almost everywhere.
  \end{enumerate}
\end{definition}

\begin{lemma}
  \label{lem:topological-transversal}
  Let $Τ$ be separable and
  \begin{enumerate}
  \item either $A_1=B_1$ or $A_1=0$; and
  \item $∂_q A_{0} + ∂_p B_{0} > 0$ almost everywhere.
  \end{enumerate}
  Then $Τ$ is topologically transverse and transverse if
  $B_1'(0) ≠ 0$.
\end{lemma}

\begin{proof}
  Let $γ$ be a homoclinic connection for $H$. Then, in the case that $A_1=B_1$,
  \begin{equation}
    \label{eq:topological-transversal-m}
    M_{γ}(ξ) = \iint_{r} \left( ∂_q A_0(q,p) + ∂_p B_0(q,p) \right)\,
    \D{p} ∧ \D{q} × B_1(ξ) = c(γ) × B_1(ξ),
  \end{equation}
  where $c(γ)>0$. Since $B_1$ is odd and not identically $0$, the
  first case of the lemma is proven; the second case is similar.
\end{proof}

\begin{example}
  \label{ex:topological-transversal-m}
  A simple example of a monotone thermostat vector field that
  satisfies the hypotheses of lemma \ref{lem:topological-transversal}
  is this: Let $F : Ξ → \R$ be a non-constant, smooth function and
  define $Τ$ by
  \begin{align}
    \notag
    G_{β} &= H + F, & A_0   &= 0,     & A_1 &= 0\\
    \label{eq:logistic-separable}
          &         & B_0   &= p,     & B_1 &= -F', & C_0 &= p · H_p  - T, & C_1 &= 1.
  \end{align}
  This defines a thermostat vector field that is separable. It is
  monotone at a saddle connection $γ$ if
  $T > \max\limits_{γ}\set{ p · H_p }$, which holds for all $T$
  sufficiently large. When $F(ξ)=1-\cos(ξ)=½ ξ^2 + O(ξ^4)$, the
  thermostat behaves, for $ξ \sim 0$, similar to the Nos{é}-Hoover
  thermostat where $F(ξ)=½ ξ^2$ \eqref{eq:nose-hoover}.
\end{example}

\section{Horseshoes in a thermostated $2$-degree of freedom hamiltonian}
\label{sec:horseshoes}

The previous section dealt with the creation of transverse
intersections on stable and unstable manifolds in a single
thermostated system using the reduced dynamics. In this section,
thermostated hamiltonians are studied using the hamiltonian formalism
on the extended phase space. The following lemma establishes the
general connection between the reduced thermostat vector field $Y_{ε}$
when the thermostat vector field $Τ$ is separable and takes a
particularly simple form.

\begin{lemma}
  \label{lem:single-2-thermostat}
  Let $Τ = -ρ F'(ξ) ∂_{ρ} + (-T+ρ · H_{ρ}) ∂_{ξ}$ be a $C^{r-1}$
  thermostat vector field for the $C^r$ hamiltonian $H$. Then, for
  $ε ≠ 0$, the vector field $Y_{ε} = X_H + ε Τ$ is a reduction of the
  hamiltonian vector field $X_{G_{ε}}$,
  \begin{equation}
    \label{eq:single-2-thermostat-reduction}
    G_{ε}(q,p,s,S) = H(q,p/s) + F_{ε}(S) + T \ln s,
  \end{equation}
  where $ε S = ξ$ and $F_{ε}(S)=F(ξ)$.
\end{lemma}
\begin{proof}
  Let $H_i$ be the partial derivative of $H$ with respect to the
  $i$-th variable. Hamilton's equations for $G_{ε}$ are
  \begin{align}
    \label{eq:single-2-thermostat-hde}
    \dot{q} & = H_2(q,p/s),  &  &  & \dot{p} & = -H_1(q,p/s),           \\\notag
    \dot{s} & = F_{ε}'(S),   &  &  & \dot{S} & = -T/s + (p/s) · H_2(q,p/s)/s.
  \end{align}
  If one introduces the new time variable $τ$ such that
  $x' = \ddt[x]{τ}=s\ddt[x]{t} = s \dot{x}$, $ρ=p/s$ and $ε S=ξ$, then these
  questions are transformed to
  \def\dat#1{{#1}'}
  \begin{align}
    \label{eq:single-2-thermostat-hde-2}
    \dat{q} & = H_2(q,ρ),    &  &  & \dat{p} & = -H_1(q,ρ)-ε\,ρ\,F'(ξ), \\\notag
    \dat{s} & = ε\,s\,F'(ξ), &  &  & \dat{ξ} & = ε\,\left( -T + ρ · H_2(q,ρ) \right).
  \end{align}
  One can eliminate the differential equation for $s$ and arrive at a
  system of differential equations described by the vector field
  $Y_{ε}$.
\end{proof}

As noted in the introduction, the Nos{é}-Hoover thermostat occupies a
privileged position in the literature on thermostat dynamics. Let us
formulate a definition which captures this centrality.

\begin{definition}
  \label{def:elementary-2-thermostat}
  A $C^r$, $r>2$, function $N_T(s,S) = F(S) + T \ln s$ is an
  \defn{elementary thermostat of order $2$} if $F(0)=F'(0)=0$,
  $F''(0)>0$ and $F'$ vanishes only at $0$.
\end{definition}

{\em Order} in this definition refers to the order of the first
non-trivial term in the Maclaurin expansion of $F$. {\em Elementary}
means that the function $N$ is a sum of two functions, each depending
on a conjugate variable. Higher-order thermostats are in the
literature, but the techniques of the present note are inapplicable.

\begin{proposition}
  \label{prop:single-2-thermostat}
  The following are elementary thermostats of order $2$:
  \begin{enumerate}
  \item the Nos{é}-Hoover thermostat~\cite{doi:10.1080/00268978400101201,nose,hoover};
  \item Tapias, Bravetti \& Sanders logistic
    thermostat~\cite{10.12921/cmst.2016.0000061,10.12921/cmst.2017.0000005};
  \item the thermostat in example~\ref{ex:topological-transversal-m}.
  \end{enumerate}
\end{proposition}
\begin{proof}
  In each case, $F(S)=½ S ² + O(S ⁴)$.
\end{proof}

\subsection{Elementary thermostats of order 2: split homoclinic connections}
\label{sec:spec}

In this section, it will be assumed that all hamiltonians are $C^r$
for $r ≥ 3$. The results that are proven below do hold for $r ≥ 2$,
but the details are somewhat more cumbersome. In addition, the
following assumptions are made:
\begin{enumerate}
\item \label{it:spec-ham-1} $H=H_m(q_1,p_1,q_2,p_2)$ is a $C^r$
  hamiltonian with $q_1$ an angle variable defined $\bmod 2 π$ and
  \begin{equation}
    \label{eq:spec-ham}
    H = \underbrace{½ p_1 ²}_{H^{(1)}} + \underbrace{\left( \frac{1}{2m} p_2 ² + V(q_2) \right)}_{H_m^{(2)}}
  \end{equation}
\item \label{it:spec-ham-2} $V$ has a non-degenerate local maximum point at $q_2=0$:
  $V(0)=0=V'(0)$ and $V''(0) = - a^2 < 0$;
\item \label{it:spec-ham-3} $Γ$ is a parameterization of a branch of
  the homoclinic loop of the saddle fixed point $(0,0)$ of $H^{(2)}:=H_1^{(2)}$,
  $Γ(t) = (q_2(t),p_2(t))$;
\item \label{it:spec-ham-4} The thermostat
  $N_T(s,S) = F(S) + T \ln s + O(S^{3})$ is an elementary thermostat
  of order 2, i.e. there is an $M>0$ such that
  $F(S) = \frac{1}{2M} S ² + O(S ³)$;
\item \label{it:spec-ham-5} The phase space of the thermostated
  hamiltonian $G=G_{1}$~\eqref{eq:single-2-thermostat-reduction} is
  $\phasespace = \cotangent (\T^1 × Ξ × \R_+)$ where $Ξ = \R$ or
  $\T^1$ is the domain of $V$.
\end{enumerate}

Let us explain some consequences of these assumptions.

\begin{lemma}
  \label{lem:nhim}
  Under the assumptions (\ref{it:spec-ham-1}--\ref{it:spec-ham-4}),
  the submanifold
  \begin{equation}
    \label{eq:nhim-unreduced}
    N = \set{(q_1,p_1,0,0,s,S)} = \cotangent \T^1 × \set{(0,0)} ×
    \cotangent \R_+,
  \end{equation}
  is a normally hyperbolic invariant manifold for the thermostated
  hamiltonian vector field
  $G=G_1$~\eqref{eq:single-2-thermostat-reduction}. Moreover, if
  $N_0 = \set{ P ∈ N \mid p_1 ≠ 0}$, then $N_0$ is foliated by
  $2$-dimensional tori which degenerate to a family of circles
  parameterized by $p_1$.
\end{lemma}

The proof of the lemma is straightforward.

\medskip

Let $p_1 = μ ≠ 0$ be fixed. Since $p_1$ is a first integral of the
thermostated vector field $X_{G}$, we can symplecticly reduce the
thermostat phase space by fixing $p_1=μ$ and ignoring the cyclic
variable $q_1$. Since the hamiltonian $G$ is invariant under the
symplectic automorphism that maps $(q_1,p_1) → (-q_1,-p_1)$ and fixes
the other coordinates, it can be assumed without loss of generality
that $μ > 0$.

If $x$ is an object on the thermostat phase space $\phasespace$ that
is $\T^1$ invariant, then $\reduction{x}_{μ}$ denotes the reduced
object. In particular, $\reduction{\phasespace}_{μ}$ is the reduced
phase space which is symplectomorphic to $\cotangent (Ξ × \R_+)$.

\begin{lemma}
  \label{lem:nhim-reduced}
  Fix $β = 1/T > 0$ and $μ > 0$. The reduced hamiltonian
  $\reduction{G}_{μ}$ is transformed to $T\, (\reduction{G}_{β,μ} +
  \reduction{R}_{β}) + ½ T \ln T$ under the transformation $s=σ/√ T$, $S=Σ √ T$
  where
  \begin{equation}
    \label{eq:nhim-reduced-hamiltonian}
    \reduction{G}_{β,μ} = \underbrace{½ \left( μ/σ \right)^2 + \ln(σ) + \dfrac{1}{2M} Σ^2 }_{\reduction{G}^{(1)}} + \underbrace{\dfrac{1}{2m} \left( p_2/σ \right)^2 + β V(q_2)}_{\reduction{G}^{(2)}},
  \end{equation}
  and $\reduction{R}_{β}(Σ) = β F(Σ/√ β) - \frac{1}{2M} Σ ²$ is a
  function of $Σ/√ β$ that vanishes to order $3$ at $Σ=0$.
\end{lemma}

It follows from the lemma that, up to a rescaling of time, the orbits
of the reduced thermostated vector field $X_{\reduction{G}_{μ}}$
coincide with those of $X_{\reduction{G}_{β,μ} + \reduction{R}_{β}}$.

\begin{lemma}
  \label{lem:nhim-reduced-hamiltonian-u-U}
  Let $1/α^2 = \sqrt{2M} μ > 0$. The hamiltonian $\reduction{G}_{β,μ}$
  is transformed to
  \begin{equation}
    \label{eq:nhim-reduced-hamiltonian-u-U}
    \reduction{G}_{β,μ} = \underbrace{½ \left( 1/(1-α u)\right)^2 + \ln(1-α u) + α ² U^2 }_{\reduction{G}^{(1)}} + \underbrace{\dfrac{1}{2m μ^2} \left( p_2/(1-α u) \right)^2 + β V(q_2)}_{\reduction{G}^{(2)}},
  \end{equation}
  and $\reduction{R}_{β}(U) = O( (U β^{-½} μ^{-1} α^{-1} )^3 )$ under the transformation $σ = μ(1-α u)$, $Σ = - U/(μ α)$.
  Thus,
  \begin{align}
    \label{eq:nhim-reduced-hamiltonian-u-U}
    \reduction{G}^{(1)} &= α ² (u^2 + U^2) + O((α u)^3), &&& \reduction{G}^{(2)} &= β V(q_2) + \dfrac{1}{2m μ^2} p_2 ² \left( 1 + 2 α u + O((α u)^2) \right).
  \end{align}
\end{lemma}

The proof is a straightforward calculation. Note that I have resisted
the temptation to put $\reduction{G}^{(1)}$ into Birkhoff normal form
because, although the remainder term on the left is improved, the
coupling term involving $u$ in $\reduction{G}^{(2)}$ becomes slightly
less transparent.

The penultimate step to obtain a simple normal form in a neighbourhood
of the saddle-centre singularity is a non-symplectic change of
variables. A scale parameter $ε>0$ is introduced via the change of
variables and energy in the first subsystem:
\begin{align}
  \label{eq:nhim-cov-1}
  (u,U) &= ε(w,W), &&& \reduction{G}^{(1)} &\mapsto ε^{-2} \reduction{G}^{(1)}.
\end{align}
And, momentum and energy are rescaled in the second subsystem:
\begin{align}
  \label{eq:nhim-cov-1}
  p_2 &= m μ^2 P_2, &&& \reduction{G}^{(2)} &\mapsto β^{-1} \reduction{G}^{(2)}.
\end{align}
Of course, the two systems are coupled and the two change of variables
\& energy cannot be decoupled as such. So, we must apply the change of
variables to the hamiltonian vector field $X_{\reduction{G}_{β,μ} +
  \reduction{R}_{β}}$, which produces a non-canonical vector field.

\begin{lemma}
  \label{lem:nhim-cov-vf}
  Let $p_2 = m μ^2 P_2$, $(u,U) = ε(w,W)$ and set
  \begin{align}
    \label{eq:nhim-parameters}
    m         & = (ε/μ)^2,                                         &  &  & β         & = m μ^2 = ε^2, &  &  & μ & = 1/(α γ).
  \end{align}
  Then, the hamiltonian differential equations of the reduced,
  rescaled thermostated hamiltonian
  $β \reduction{G}_{μ} = \reduction{G}_{β,μ} + \reduction{R}_{β}$ are
  transformed to
  \begin{align}
    \label{eq:nhim-cov-des}
    \dot{W}   & = -2 α^2 w - α ε \left( P_2^2 + 5 (α w)^2 \right), &  &  & \dot{w}   & = 2 α^2 W - b γ^3 W^2/2, \\\notag
    \dot{P}_2 & = -V'(q_2),                                        &  &  & \dot{q}_2 & = P_2 \left( 1 + 2 α ε w \right)
  \end{align}
  modulo $O(ε^2) + o(γ^3)$ where $b=F'''(0)$.
\end{lemma}

From lemma~\ref{lem:nhim-reduced-hamiltonian-u-U} it follows that the
parameter $γ = \sqrt[4]{2M/μ^2}$. So, we can treat $ε, α$ and $γ$ as
independent positive parameters in~\eqref{eq:nhim-cov-des} while $b$
is a fixed constant. Note that the non-linear terms in the right-hand
side of $\dot{w}$ originate from the remainder term $R_{β}$--the
rescaling is unable to remove these terms. Recall that the plane
$\reduction{N}_{0,μ} = \set{(w,W,q_2,P_2) \mid q_2=P_2=0}$ is a
normally hyperbolic invariant manifold for all non-negative values of
the parameters. When $ε=0$, the two pairs of differential equations
decouple into a (non-linear, hamiltonian) oscillator in the $(w,W)$
plane and a mechanical hamiltonian in the $(q_2,P_2)$ plane with a
saddle at $(0,0)$ and a saddle connecting orbit
$Γ(t) = (q_2(t), P_2(t))$ (this was the motivation to fix $β$ in terms
of $m μ^2$). The hamiltonian function
$H^{(2)}_1(q_2,P_2)$~\eqref{eq:spec-ham}, which will be denoted by
$H^{(2)}$, is a constant of motion and $H^{(2)} \circ Γ ≡ 0$ by the
hypotheses on $V$. On the other hand, when $ε > 0$, the local stable
and unstable manifolds of $\reduction{N}_{0,μ}$,
$\reduction{W}^{±}_{ε,loc}$, are graphs over
$\reduction{W}^{±}_{0,loc} ⊃ \reduction{N}_{0,μ} × Γ(\R_{±})$. In
general, these invariant manifolds do not coincide, but intersect
transversely along homoclinic orbits. If
$Ω^{±}_{ε} : \reduction{W}^{±}_{0,loc} → \reduction{W}^{±}_{ε,loc}$ is
a parameterization of the local stable and unstable manifolds that is
$C^2$ in all variables and $Ω^{±}_{0}$ is the identity map, then
$d_{ε} = H^{(2)} \circ Ω^{+}_{ε} - H^{(2)} \circ Ω^{-}_{ε} = ε M + O(ε
²)$ measures the distance between the invariant manifolds. The
\poincaremelnikov{} function $M : \reduction{W}^{±}_{0,loc} → \R$
measures the lowest order difference in the invariant manifolds. If
$M$ vanishes at a point $P$ and $\D{M}_P ≠ 0$, then the implicit
function theorem implies that the zero set of $M$ is a surface in a
neighbourhood of $P$. In addition, the implicit function theorem
implies that the zero set of $d_{ε}$ is an $O(ε)$ perturbation of the
zero set of $M$ near $P$. Since $\reduction{W}^{±}_{ε,loc}$ are
codimension-1 submanifolds, the zero set of $d_{ε}$ is where they
intersect, and if $\D{M}_P ≠ 0$, they intersect transversely near
$P$.

To evaluate the \poincaremelnikov{} integral for the present problem,
let $φ^t_{ε,γ}$ be the flow mapping of the vector field
$Z_{ε,γ} = Z_0 + ε Z_1 + O(ε^2,γ^3)$ defined by the differential
equations~\eqref{eq:nhim-cov-des}. For $P=(w,W,q_2(t_0),P_2(t_0))$ in
$\reduction{N}_{0,μ} × Γ(\R_{±}) ⊂ \reduction{W}^{±}_{0,loc}$, the
\poincaremelnikov{} function at $P$ is
\begin{equation}
  \label{eq:nhim-melnikov-fn}
  M(P) = \int\limits_{-∞}^{∞} \kp{\D{H}^{(2)}}{Z_1}\, \circ \, φ^t_{0,0}(P) \, \D{t}.
\end{equation}

To compute the \poincaremelnikov{} integral, let us recall the
\defn{Fourier transform} $\fouriertransform{f}$ of an integrable
function $f ∈ L^1(\R)$~\cite{MR0304972,MR1970295}:
\begin{equation}
  \label{eq:nhim-fourier-transform}
  \fouriertransform{f}(x) = \int\limits_{\R} f(t)\,\exp(-2 π i xt)\, \D{t}.
\end{equation}
If there are positive real constants $A, λ$ such that $|f(t)| ≤ A
\exp(-λ |t|)$ for all real $t$, then $\fouriertransform{f}$ is a
holomorphic function of $x$ in the strip $|\imagpart{x}| < λ/2 π$.
This implies that $\fouriertransform{f}$ has at most countably many
real zeros.

Recall the standard symplectic form $Ω$ on $\C^n$: for $n=1$, $Ω(z',z)
= \imagpart{z' \bar{z}}$ and for $n>1$, $Ω$ is a sum of $n$ copies of
the elementary form. In components, if $z'=x'+iy'$, $z=x+iy$ and
$x',x,y',y$ are all real, then $Ω(z',z) = y'x - x'y$.

\begin{proposition}
  \label{prop:nhim-melnikov-integral}
  Let $f(t)=P_2(t)^2$ where $Γ(t)=(q_2(t),P_2(t))$ is the homoclinic
  solution from assumption~\ref{it:spec-ham-3}. Then, the
  \poincaremelnikov{} function is the real-linear function
  \begin{equation}
    \label{eq:nhim-melnikov-fn-calc}
    M(P) = \dfrac{1}{2 α} \, Ω(\exp(2 i α^2 t_0)z_0, \fouriertransform{f}(α^2/π))
  \end{equation}
  where $P ∈ \reduction{N}_{0,μ} × Γ(\R) ⊂ \reduction{W}^{±}_{0,loc}$,
  $P = (w,W,q_2,P_2)$ and $z_0=w+iW$, $(q_2,P_2)=Γ(-t_0)$.
\end{proposition}
\begin{proof}
  Let
  $P=(w,W,q_2,P_2) ∈ \reduction{N}_{0,μ} × Γ(\R) ⊂
  \reduction{W}^{±}_{0,loc}$. There is a unique $t_0$ such that
  $(q_2,P_2)=Γ(-t_0)=(q_2(-t_0),P_2(-t_0))$. The flow mapping
  $φ^t_{0,0}$ maps
  $(z_0,Γ(-t_0)) → (z(t)=\exp(2 i α^2 t) z_0, Γ(t-t_0))$. Thus,
  \begin{align}
    \label{eq:nhim-melnikov-integral}
    M(P) &= \int\limits_{\R} 2 α w(t) P_2(t-t_0) V'(q_2(t-t_0))\, \D{t}, & \text{from~(eq.s \ref{eq:nhim-cov-des}, \ref{eq:nhim-melnikov-fn})} \\\notag
         &= \frac{1}{2 α}\, \int\limits_{\R} W(t) P_2(t-t_0)^2 \, \D{t},& \text{from~\eqref{eq:nhim-cov-des}}\\\notag
         &= \frac{1}{2 α}\, \int\limits_{\R} \imagpart{z(s+t_0)} P_2(s)^2 \, \D{s},& \text{since $P_2$ is real}
  \end{align}
  which yields~\eqref{eq:nhim-melnikov-fn-calc} from the definition of
  the Fourier transform and the symplectic form $Ω$.
\end{proof}
\begin{remark}
  \label{rem:nhim-melnikov-fn-calc}
  The \poincaremelnikov{} function is invariant under the unperturbed
  flow. This follows, in the current case, from elementary properties
  of the Fourier transform.

  A more careful statement of Lemma~\ref{lem:nhim-cov-vf} shows that
  for $ε=0$, the additional scaling parameter $γ$ does not destroy the
  integrability of the differential
  equations~\eqref{eq:nhim-cov-des}. It is possible to determine the
  \poincaremelnikov{} integral for $γ>0$, but since it coincides up to
  third-order in $γ$ with $M$~\eqref{eq:nhim-melnikov-fn-calc}, the
  computation has been elided.
  
  The calculation and result in
  Proposition~\ref{prop:nhim-melnikov-integral} is similar to that
  in~\cite[\S V.3.9]{MR1411677}. In the latter work, Kozlov studies
  the effect of a sinusoidal, time-dependent forcing of an ideal
  stationary planar perfect fluid flow. If the unperturbed systems has
  a saddle fixed point, the $1 ½$-degree of freedom forced system has
  hyperbolic periodic orbits and the \poincaremelnikov{} integral for
  the specific perturbation leads to a Fourier transform.

  The splitting of invariant manifolds in the hamiltonian setting is a
  well-studied problem. Koltsova and Lerman~\cite{MR1316697,MR1409407}
  prove that for a {\em generic} hamiltonian system with a
  saddle-centre critical point, at all nearby positive energy levels,
  there is a hyperbolic periodic orbit with 4 transverse homoclinic
  orbits. This can be explained in a simple way: in appropriate
  coordinates like those used in the proof of
  Proposition~\ref{prop:nhim-melnikov-integral}, the
  \poincaremelnikov{} function is, to lowest order, an indefinite
  quadratic form in $(w,W)$ which vanishes at the origin (the
  saddle-centre). The 4 transverse homoclinics originate from the 4
  half-lines emanating from the origin where the quadratic
  vanishes. In the present case, the zeros of the \poincaremelnikov{}
  function occur on 2 half-lines emanating from the origin of the
  $(w,W)$ plane; it follows that the splitting studied here is not
  generic in the sense of Koltsova and Lerman, and that the reduced
  hamiltonians are not generic.

  On the other hand, Churchill, Pecelli and Rod~\cite{MR569596} study
  the stability of the periodic orbits in the $(q_2,P_2)$ plane which
  limit onto the separatrix $Γ$ from energy levels below $0$. They
  show that, under mild conditions satisfied by the H{é}non-Heiles
  hamiltonian for example, there will be an infinite sequence of
  intervals of energy converging to $0$ where the periodic orbits are
  alternately elliptic and hyperbolic. However, in the current
  problem, one sees that aside from $ε=0$, there are appear to be no
  such continuous family of periodic orbits.
\end{remark}

\begin{theorem}
  \label{thm:nhim-main-thm}
  Assume hypotheses~(\ref{it:spec-ham-1}--\ref{it:spec-ham-5}). Assume
  that $\fouriertransform{f}$ of
  Proposition~\ref{prop:nhim-melnikov-integral} does not vanish at
  $α ²/π$. Then, there exists an $ε_0=ε_0(α)$ such that for all
  $0 < ε, γ < ε_0(α)$, the stable and unstable manifolds of
  $\reduction{N}_{0,μ}$, $\reduction{W}^{±}_{ε,γ}$, intersect
  transversely.

  In particular, in a neighbourhood of the saddle-centre critical
  point, the differential equations possess an invariant horseshoe.
\end{theorem}

\begin{corollary}
  \label{cor:nhim-main-thm}
  Assume hypotheses~(\ref{it:spec-ham-1}--\ref{it:spec-ham-5}). Let
  $V(q)$ equal:
  \begin{align*}
    (1)\ \ & -½ q ² (1 - q ⁿ); &&& \text{ or }(2)\ \ & \cos(q)-1,
  \end{align*}
  where $n$ is a positive integer. Then $ε_0(α)>0$ for all but one
  value of $α$ in case (1), and all $α$ in case (2).
\end{corollary}
\begin{proof}
  In case (1), to compute the Fourier transform of $p ²$, let
  $z=p/q$. Along the homoclinic connection $Γ$ of the saddle at
  $(0,0)$, one has
  \begin{align}
    \label{eq:nhim-main-thm-cor-1}
    z ² + q ⁿ &= 1, &&& 2 \D{z} + n q ⁿ \D{t} &=  2 \D{z} + n \left( 1 - z ² \right)\, \D{t} = 0,\\\notag
    \ln \left( \dfrac{1-z}{1+z} \right) &= n \left( t - t_0 \right).
  \end{align}
  Since $p ² = z ² q ²$, if one chooses $q=1,p=0$ at $t_0=0$, then one obtains that (with $ω=2 π ξ$)
  \begin{align}
    \label{eq:nhim-main-thm-cor-2}
    \fouriertransform{p ²}(ω) &= \dfrac{2}{π n} × \realpart{\int_0^1 z^{b-1}\, \left( 1-z \right)^{c-b-1}\, \left( 1+z\right)^{-a}\, \D{z}}\\
    &= \dfrac{4}{π n} × \realpart{\dfrac{\hypergeometric{a,b}{c}{-1}}{(c-1)(c-2)(c-3)}} &&& \text{by~\cite[15.3.1]{MR1225604}}
  \end{align}
  where $a=1-2/n-i ω/n$, $b=3$, $c=b+1-\bar{a}$ and
  $\hypergeometric{a,b}{c}{z}$ is the Gaussian hyper-geometric function
  defined for $\realpart{c}>\realpart{b}>0$ and
  $z ∈ \C \setminus [1,∞)$. For the case when $n=4/(k-2)$ where $k>2$
  is an integer, the Fourier transform can also be evaluated using the
  method of residues, which results in
  $\fouriertransform{p ²}(ω) = g_k(ω) × G_k(ω)$ where $g_k(2x/(k-2))$
  is a product of a quadratic polynomial with exactly two real roots
  $x= ±\sqrt{k-2}$ and, when $k=2m-1$ is odd the second factor is
  $\prod_{l=1}^{m-2} (x^2 + (2l-1)^2)$, otherwise when $k=2m$ is even,
  the second factor is $\prod_{l=1}^{m-2} (x^2 + 4l^2)$; and
  $G_k(2x/(k-2))$ equals $\sech{x π/2}$ when $k$ is odd, and
  $x \cosech{x π/2}$ when $k$ is even. Selected examples are plotted
  in figure~\ref{fig:fourier-transform}.

  In case (2), one computes that $p ² = 4 \sech{t}^2$ along a
  homoclinic connection. The Fourier transform evaluates to
  $\fouriertransform{p ²}(ω) = 2 ω \cosech{ω π/2}$, which is positive
  everywhere. As an amusing side-note, the substitution $z=\cos(q/2)$
  leads to an integral in the form of the right-hand side
  of~\eqref{eq:nhim-main-thm-cor-2} with $n=1/2$, $a=-i ω/2$, $b=1$
  and $c=2+a$. This gives $\fouriertransform{p ²}(2 ω)$ is the real
  part of $4 \hypergeometric{-i ω,1}{2-i ω}{-1}/π(1-i ω)$.
\end{proof}
\begin{figure}[h]
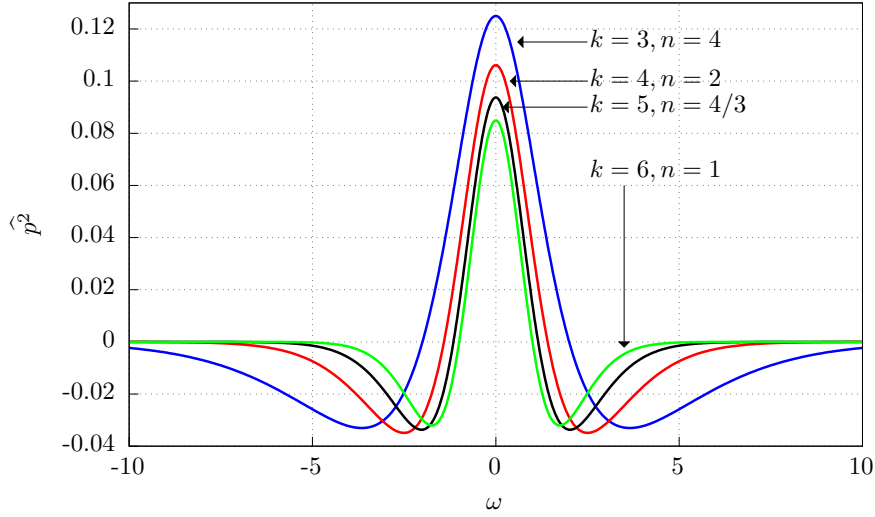

  \label{fig:fourier-transform}
  \caption{The Fourier transform of $p ²$ vs. $ω=2 π ξ$ for selected values of $n$.}
  \ltxfigure{nhht-nhim-main-thm-mac.ltx}{12cm}{!}
\end{figure}

\begin{octavecode}
function r=p2hat(w,k)
n=4/(k-2);
g=(2+i*w)/n;
a=1-conj(g);
b=3;
c=3+g;
r=4/pi/n*(1/((c-3)*(c-2)*(c-1)))*gsl_sf_hyperg_2F1(a,b,c,-1);
endfunction
\end{octavecode}

%% code for the sec:spec examples

%% WINKLER-like thermostat %%

%% code for the potentials

\subsection{Variable-mass thermostats}
\label{sec:winkler}

Let us extend definition~\ref{def:elementary-2-thermostat} and remove
the ``elementary'' aspect: we will allow the function $F$ to be
weighted by a variable mass, viz.

\begin{definition}
  \label{def:order-2-thermostat}
  A $C^r$, $r>2$, function $N_T(s,S) = Ω(s) F(S) + T \ln s$ is a
  \defn{variable-mass thermostat of order $2$} if $Ω>0$,
  $F(0)=F'(0)=0$, $F''(0)>0$ and $F'$ vanishes only at $0$.
\end{definition}

In place of assumption~\ref{it:spec-ham-4} of section~\ref{sec:spec},
it will be assumed henceforth

\begin{enumerate}
\item[(4)] \label{it:spec-ham-4-vm} The thermostat $N_T(s,S) = Ω(s)
  F(S) + T \ln s$ is a variable-mass thermostat of order $2$.
\end{enumerate}

Thermostats like these have been studied in the
literature~\cite{scopus2-s2.0-0001730129,scopus2-s2.0-4243944612,proquest1859245445}
in different guises. For example, one can obtain this form from
Winkler's thermostat, which rescales momenta by $p/s^e$ ($e=2$ in
Winkler's case), by the transformation
$(s,S) = (s_1^{1/e}, es_1^{1/r} S_1)$ where $1/e+1/r=1$ and
$1 < e, r < ∞$. In the $(s_1,S_1)$ variables, the thermostat rescales
momentum by $p/s_1$ and the effective thermostat temperature is
$T_1=T/e$. The Nos{é}-Hoover thermostat $N(s,S) = ½ S ² + T \ln(s)$ is
transformed to $N(s_1,S_1) = Ω(s_1) F(S_1) + T_1 \, \ln s_1$ where
$Ω(s_1) = e ² s_1^{2/r}$ and $F(S_1) = ½ S_1 ²$ (so $Ω(s_1) = 4 s_1$
for Winkler's thermostat).

Let us state a theorem in the case when $Ω(s)$ is regular at $s=0$:

\begin{theorem}
  \label{thm:winkler-a1-pos-a1=0}
  Assume hypotheses~\ref{it:spec-ham-1}--\ref{it:spec-ham-5}.
  Assume that $Ω, F$ are $C^3$ and $Ω(s) = a_1 + b_1 s + c_1 s ²/2 +
  O(s ³)$, $F(S) = ½ S ² + ⅓ b S ³ + o(S ⁴)$. Then,
  \begin{enumerate}
  \item if $a_1>0$, then the conclusions of
    Proposition~\ref{prop:nhim-melnikov-integral} and
    Theorem~\ref{thm:nhim-main-thm} hold;
  \item if $a_1=0$, $b_1>0$ and $c_1 ≥ 0$ then under the substitution
    \begin{align}
      \label{eq:winkler-a1=0}
      α^2       & = \sqrt{β μ c_1+2 \sqrt{β} b_1}/(2 \sqrt{μ}),  &  &  & m         & = ε^2/(μ^2),                              \\\notag
      β         & = ε^2,                                         &  &  & μ         & = 1/αγ                                    \\
      s         & = μ\sqrt{β}(1-α ε w),                          &  &  & S         & = -ε W/(\sqrt{β} μ α)
    \end{align}
    the hamiltonian differential equations of the rescaled, reduced
    thermostated hamiltonian $β\, \reduction{G}_{μ}$ are transformed
    to
    \begin{align}
      \label{eq:winkler-a1=0-des}
      \dot{W}   & = -2 α ² w                                     &  &  & \dot{w}   & = 2 α ² W                                 \\\notag
                & \phantom{=} + α ε ((α W)^2 -5 (α w)^2 -P_2^2), &  &  &           & \phantom{=}  - γ b α ² W ² - 2 ε α ³ w W, \\\notag
      \dot{P}_2 & = -V'(q_2)                                     &  &  & \dot{q}_2 & = P_2 \left( 1+2 α ε w \right)
    \end{align}
    modulo $O(ε ², γ ², ε γ)$. In particular, the conclusions of
    Proposition~\ref{prop:nhim-melnikov-integral} and
    Theorem~\ref{thm:nhim-main-thm} hold.
  \end{enumerate}
\end{theorem}

\begin{remark}
  \label{rem:winkler-a1-pos-a1=0}
  The proof of the theorem is straightforward and follows essentially
  the same arguments as above. Let us explain the differences
  between~\eqref{eq:winkler-a1=0-des} and~\eqref{eq:nhim-cov-des}. In
  \eqref{eq:winkler-a1=0}, one can rewrite the equation for $α ²$ as
  $b_1 = 2 α ⁴ μ β^{-½} - c_1 μ β^{½}$, which explains the appearance
  of the $(α W) ²$ term in the equation for
  $\dot{W}$~\eqref{eq:winkler-a1=0-des}. It also explains why the term
  that multiplies $b$ in the $\dot{w}$ equation has a factor of $γ$
  vs. $γ ³$ and the appearance of the $wW$ term in the same
  equation. It should also be noted that when $a_1=0$, the mass-like
  term $1/(2b_1)$ tends to $0$ as $β,1/μ → 0$ and $α > 0$ is
  fixed. This is similar to the constant mass case where $M → 0$ under
  the same conditions.
\end{remark}

\section{Conclusion}
\label{sec:conclusion}

It has been shown that the Nos{é}-Hoover thermostat and
closely-related thermostats create transverse homoclinic points near
the saddle-centre equilibria of the reduced differential equations
when applied to a separable system that is a sum of a $1$-dimensional
ideal gas and a planar pendulum. Several extensions of this example
have been described. Several problems remain, including:
\begin{enumerate}
\item Extend the present results to thermostats like those in
  \cite{sciversesciencedirect_elsevier0375-9601.95.00973-6,10.1016/j.cnsns.2015.08.02,PhysRevE.75.040102}
  where the thermostat $N$ depends on $p,s,S$ and/or the lowest-order
  term in $N$ is $S^4$ (or higher);
\item Extend the present results to multi-variable thermostats, such as
  those discussed in the introduction;
\item Demonstrate the existence of transverse homoclinic points in
  systems like the thermostated harmonic oscillator;
\item Demonstrate the existence of Arnol'd diffusion in Nos{é}-Hoover
  thermostated $n ≥ 2$-degree of freedom systems~\cite{MR0163026}.
\end{enumerate}
The last problem motivated the present paper. However, it turns out
that proving the existence of transverse homoclinic points in a
neighbourhood of the periodic orbits of mixed type is already a
sufficiently rich problem. The third problem also partially motivated
the present paper: in contrast, though, its solution will involve
exponentially small splitting and the more delicate calculations and
error estimates that are involved.

\section*{Acknowledgements}
\label{sec:ack}

This research has been partially supported by the Natural Science and
Engineering Research Council of Canada grant 320 852.

\bibliographystyle{abbrv}
\bibliography{nhht-references}
\end{document}

%% file: definitions.tex
\usepackage[citebordercolor={0.7 0.7 0.7},linkbordercolor={0.7 0.7 1}]{hyperref}
\usepackage{comment,graphicx,color}
\usepackage[all]{xy}
\excludecomment{maximacode}
\excludecomment{octavecode}
\excludecomment{arxivabstract}
\includecomment{arma}
\usepackage{ltbunicode}

\ifcsname showkeystrue\endcsname
\usepackage[notcite,notref]{showkeys}

\fi

\DeclareGraphicsExtensions{.png,jpg,jpeg,.mps,.pdfmac,.pdffig,.spdf,.pdf}
\newcommand{\T}{{\bf T}}
\newcommand{\R}{{\bf R}}
\newcommand{\C}{{\bf C}}

\newcommand{\set}[1]{\left\{#1\right\}}

\newtheorem{theorem}{Theorem}[section]
\newtheorem{corollary}{Corollary}[section]

\newtheorem{lemma}{Lemma}[section]
\newtheorem{proposition}{Proposition}[section]

\theoremstyle{definition}
\newtheorem{definition}{Definition}[section]
\theoremstyle{remark}

\newtheorem{remark}{Remark}[section]
\newtheorem{example}{Example}[section]

\newcommand{\safeinput}[1]{\IfFileExists{#1}{{\catcode`\&=0\input{#1}}}{\message{File
        #1 does not exists. Skipping.}}}
\def\ltxfigure#1#2#3{\resizebox{#2}{#3}{\input{#1}}}

\newcommand{\D}[1]{\mathrm{d}#1}

\newcommand{\textem}[1]{{\em #1}}
\newcommand{\defn}[1]{\textem{#1}}

\ifcsname nhpaper\endcsname
\else
\fi

\newcommand{\temp}[1]{2{\mathrm K}}

\newcommand{\kp}[2]{\langle #1, #2 \rangle}
\newcommand{\ddt}[2][\ ]{\frac{\D{#1}}{\D{#2}}}

\newcommand{\cotangent}{T^*}
\newcommand{\metaref}[3]{{\ifnum#1=0(\fi}{#3}\ref{#2}{\ifnum#1=0)\fi}}
\renewcommand{\eqref}[2][0]{\metaref{#1}{#2}{eq.~}}

\newcommand{\pb}[2]{\left\{ #1,#2 \right\}}

\newcommand{\phasespace}[0]{{\mathcal P}}
\newcommand{\reduction}[1]{\hat{#1}}
\newcommand{\fouriertransform}[1]{\widehat{#1}}
\newcommand{\realpart}[1]{\mathbf{Re}\left(#1\right)}
\newcommand{\imagpart}[1]{\mathbf{Im}\left(#1\right)}
\newcommand{\hypergeometric}[3]{\,F\hspace{-.1em}\left(#1;#2;#3\right)}
\newcommand{\sech}[1]{\,{\rm sech}(#1)}
\newcommand{\cosech}[1]{\,{\rm cosech}(#1)}
\newcommand{\poincaremelnikov}[0]{Poincar{é}-Melnikov}
\long\def\aureply#1{\ifhmode\newline\fi\noindent{\bf Author Reply}:\ {\em #1}}

\def\gettimestamp#1<#2>{\def\timestamp{\tt #2}}

%% file: nhht.bbl
\def\cprime{$'$} \def\cprime{$'$} \def\cprime{$'$} \def\cprime{$'$}
\begin{thebibliography}{10}

\bibitem{MR1225604}
M.~Abramowitz and I.~A. Stegun, editors.
\newblock {\em Handbook of mathematical functions with formulas, graphs, and
  mathematical tables}.
\newblock Dover Publications Inc., New York, 1992.
\newblock Reprint of the 1972 edition.

\bibitem{andersen}
H.~C. Andersen.
\newblock Molecular dynamics simulations at constant pressure and/or
  temperature.
\newblock {\em J. Chem. Phys.}, 72:2384--2393, 1980.

\bibitem{MR0163026}
V.~I. Arnol\cprime~d.
\newblock Instability of dynamical systems with many degrees of freedom.
\newblock {\em Dokl. Akad. Nauk SSSR}, 156:9--12, 1964.

\bibitem{MR1346373}
K.~Burns and H.~Weiss.
\newblock A geometric criterion for positive topological entropy.
\newblock {\em Comm. Math. Phys.}, 172(1):95--118, 1995.

\bibitem{1806.10198v3}
{Butler, Leo T.}
\newblock Invariant tori for a class of singly thermostated hamiltonians.
\newblock 2019.

\bibitem{MR569596}
R.~C. Churchill, G.~Pecelli, and D.~L. Rod.
\newblock Stability transitions for periodic orbits in {H}amiltonian systems.
\newblock {\em Arch. Rational Mech. Anal.}, 73(4):313--347, 1980.

\bibitem{MR0501173}
M.~W. Hirsch, C.~C. Pugh, and M.~Shub.
\newblock {\em Invariant manifolds}.
\newblock Lecture Notes in Mathematics, Vol. 583. Springer-Verlag, Berlin,
  1977.

\bibitem{hoover}
W.~G. Hoover.
\newblock Canonical dynamics: equilibrium phase space distributions.
\newblock {\em Phys. Rev. A.}, 31:1695--1697, 1985.

\bibitem{sciversesciencedirect_elsevier0375-9601.95.00973-6}
W.~G. Hoover and B.~L. Holian.
\newblock Kinetic moments method for the canonical ensemble distribution.
\newblock {\em Physics Letters A}, 211(5):253,257, 1996.

\bibitem{10.12921/cmst.2017.0000005}
W.~G. Hoover and C.~G. Hoover.
\newblock {S}ingly-{T}hermostated {E}rgodicity in {G}ibbs’ {C}anonical
  {E}nsemble and the 2016 {I}an {S}nook {P}rize {A}ward.
\newblock {\em CMST}, 23(1):5--8, 2017.

\bibitem{10.1016/j.cnsns.2015.08.02}
W.~G. Hoover, J.~C. Sprott, and C.~G. Hoover.
\newblock Ergodicity of a singly-thermostated harmonic oscillator.
\newblock {\em Communications in Nonlinear Science and Numerical Simulation},
  32(Supplement C):234 -- 240, 2016.

\bibitem{scopus2-s2.0-0001730129}
J.~Jellinek and R.~Berry.
\newblock Generalization of {N}osé's isothermal molecular dynamics.
\newblock {\em Physical Review A}, 38(6):3069,3072, 1988.

\bibitem{scopus2-s2.0-4243944612}
J.~Jellinek and R.~Berry.
\newblock Generalization of {N}osé's isothermal molecular dynamics: Necessary
  and sufficient conditions of dynamical simulations of statistical ensembles.
\newblock {\em Physical Review A}, 40(5):2816,2818, 1989.

\bibitem{MR573822}
A.~Katok.
\newblock Lyapunov exponents, entropy and periodic orbits for diffeomorphisms.
\newblock {\em Inst. Hautes \'{E}tudes Sci. Publ. Math.}, (51):137--173, 1980.

\bibitem{MR1316697}
O.~Y. Kol\cprime~tsova and L.~M. Lerman.
\newblock Dynamics and bifurcations in two-parameter unfolding of a
  {H}amiltonian system with a homoclinic orbit to a saddle-center.
\newblock In {\em Hamiltonian mechanics ({T}oru\'{n}, 1993)}, volume 331 of
  {\em NATO Adv. Sci. Inst. Ser. B Phys.}, pages 385--390. Plenum, New York,
  1994.

\bibitem{MR1409407}
O.~Y. Koltsova and L.~M. Lerman.
\newblock Families of transverse {P}oincar\'{e} homoclinic orbits in
  {$2n$}-dimensional {H}amiltonian systems close to the system with a loop to a
  saddle-center.
\newblock {\em Internat. J. Bifur. Chaos Appl. Sci. Engrg.}, 6(6):991--1006,
  1996.

\bibitem{MR1411677}
V.~V. Kozlov.
\newblock {\em Symmetries, topology and resonances in {H}amiltonian mechanics},
  volume~31 of {\em Ergebnisse der Mathematik und ihrer Grenzgebiete (3)
  [Results in Mathematics and Related Areas (3)]}.
\newblock Springer-Verlag, Berlin, 1996.
\newblock Translated from the Russian manuscript by S. V. Bolotin, D. Treshchev
  and Yuri Fedorov.

\bibitem{MR1150101}
D.~Kusnezov and A.~Bulgac.
\newblock Canonical ensembles from chaos. {II}. {C}onstrained dynamical
  systems.
\newblock {\em Ann. Physics}, 214(1):180--218, 1992.

\bibitem{MR1079786}
D.~Kusnezov, A.~Bulgac, and W.~Bauer.
\newblock Canonical ensembles from chaos.
\newblock {\em Ann. Physics}, 204(1):155--185, 1990.

\bibitem{MR2299758}
F.~Legoll, M.~Luskin, and R.~Moeckel.
\newblock Non-ergodicity of the {N}os\'e-{H}oover thermostatted harmonic
  oscillator.
\newblock {\em Arch. Ration. Mech. Anal.}, 184(3):449--463, 2007.

\bibitem{MR2519685}
F.~Legoll, M.~Luskin, and R.~Moeckel.
\newblock Non-ergodicity of {N}os\'e-{H}oover dynamics.
\newblock {\em Nonlinearity}, 22(7):1673--1694, 2009.

\bibitem{MR2136523}
B.~J. Leimkuhler and C.~R. Sweet.
\newblock A {H}amiltonian formulation for recursive multiple thermostats in a
  common timescale.
\newblock {\em SIAM J. Appl. Dyn. Syst.}, 4(1):187--216 (electronic), 2005.

\bibitem{Mahdi20111348}
A.~Mahdi and C.~Valls.
\newblock Integrability of the {N}osé–{H}oover equation.
\newblock {\em Journal of Geometry and Physics}, 61(8):1348 -- 1352, 2011.

\bibitem{nose-hoover-chains}
G.~Martyna, M.~Klein, and M.~Tuckerman.
\newblock Nos\'e-{H}oover chains: The canonical ensemble via continuous
  dynamics.
\newblock {\em J. Chem. Phys.}, 97:2635--2643, 1992.

\bibitem{tayfranc10.1080/00268971003689923}
T.~Morishita.
\newblock From {N}osé–{H}oover chain to {N}osé–{H}oover network: design
  of non-hamiltonian equations of motion for molecular-dynamics with multiple
  thermostats.
\newblock {\em Molecular Physics}, 108(10):1337,1347, 2010-05-20.

\bibitem{nose}
S.~Nos\'e.
\newblock A unified formulation of the constant temperature molecular dynamics
  method.
\newblock {\em J. Chem. Phys.}, 81:511--519, 1984.

\bibitem{doi:10.1080/00268978400101201}
S.~Nosé.
\newblock A molecular dynamics method for simulations in the canonical
  ensemble.
\newblock {\em Molecular Physics}, 52(2):255--268, 1984.

\bibitem{doi:10.1142/S0218127416501704}
P.~C. Rech.
\newblock Quasiperiodicity and chaos in a generalized {N}osé–{H}oover
  oscillator.
\newblock {\em International Journal of Bifurcation and Chaos}, 26(10):1650170,
  2016.

\bibitem{MR1970295}
E.~M. Stein and R.~Shakarchi.
\newblock {\em Fourier analysis}, volume~1 of {\em Princeton Lectures in
  Analysis}.
\newblock Princeton University Press, Princeton, NJ, 2003.
\newblock An introduction.

\bibitem{MR0304972}
E.~M. Stein and G.~Weiss.
\newblock {\em Introduction to {F}ourier analysis on {E}uclidean spaces}.
\newblock Princeton University Press, Princeton, N.J., 1971.
\newblock Princeton Mathematical Series, No. 32.

\bibitem{10.12921/cmst.2016.0000061}
D.~Tapias, A.~Bravetti, and D.~P. Sanders.
\newblock Ergodicity of one-dimensional systems coupled to the logistic
  thermostat.
\newblock {\em CMST}, 23, 11 2016.

\bibitem{Wang2015}
L.~Wang and X.-S. Yang.
\newblock The invariant tori of knot type and the interlinked invariant tori in
  the {N}os{\'e}-{H}oover oscillator.
\newblock {\em The European Physical Journal B}, 88(3):78, Mar 2015.

\bibitem{doi:10.1063/1.4937167}
L.~Wang and X.-S. Yang.
\newblock A vast amount of various invariant tori in the {N}osé-{H}oover
  oscillator.
\newblock {\em Chaos: An Interdisciplinary Journal of Nonlinear Science},
  25(12):123110, 2015.

\bibitem{doi:10.1142/S0218127417501115}
L.~Wang and X.-S. Yang.
\newblock The coexistence of invariant tori and topological horseshoe in a
  generalized {N}osé–{H}oover oscillator.
\newblock {\em International Journal of Bifurcation and Chaos}, 27(07):1750111,
  2017.

\bibitem{PhysRevE.75.040102}
H.~Watanabe and H.~Kobayashi.
\newblock Ergodicity of a thermostat family of the {N}os\'e-{H}oover type.
\newblock {\em Phys. Rev. E}, 75:040102, Apr 2007.

\bibitem{proquest1859245445}
Winkler.
\newblock Extended-phase-space isothermal molecular dynamics: Canonical
  harmonic oscillator.
\newblock {\em Physical review. A, Atomic, molecular, and optical physics},
  45(4):2250,2255, 1992-02-15.

\end{thebibliography}
